\documentclass[11pt]{article}

\usepackage[a4paper,margin=1in]{geometry}
\usepackage[T1]{fontenc}
\usepackage{amsmath,amssymb,amsthm}
\usepackage{microtype}
\usepackage[colorlinks=true,linkcolor=blue,citecolor=blue,urlcolor=blue]{hyperref}

\hypersetup{
  pdftitle={Hamilton Cycles in 10-Tough (2P2 union P1)-Free Graphs},
  pdfauthor={Qiuyu Chen}
}

\newtheorem{theorem}{Theorem}[section]
\newtheorem{proposition}[theorem]{Proposition}
\newtheorem{corollary}[theorem]{Corollary}
\newtheorem{lemma}[theorem]{Lemma}
\theoremstyle{definition}
\newtheorem{definition}[theorem]{Definition}

\title{Hamilton Cycles in 10-Tough \texorpdfstring{$(2P_2\cup P_1)$}{(2P2 union P1)}-Free Graphs}
\author{Qiuyu Chen\\
Shanghai Jiao Tong University\\
\texttt{canghaimeng@sjtu.edu.cn}}
\date{}

\begin{document}

\maketitle

\begin{abstract}
A graph is called $10$-tough and $(2P_2\cup P_1)$-free if every vertex set whose deletion leaves at least two components has cardinality at least ten times the number of those components and if the graph has no induced subgraph consisting of two disjoint edges and an isolated vertex. We prove that every finite simple $10$-tough $(2P_2\cup P_1)$-free graph on at least three vertices is Hamiltonian. The proof splits according to whether some edge has joint neighbourhood of order at most $4n/11$. In the small-neighbourhood case, a matched path-cover is compressed to a prescribed matching. In the large-neighbourhood case, an asymmetric analysis of the two components left by a putative small cut yields the required connectivity bound. A Hamilton cycle through the prescribed edges is then expanded, and a cycle-extension lemma inserts the remaining vertices.
\end{abstract}

\section{INTRODUCTION}\label{sec:intro}

All graphs in this paper are finite and simple. For a finite set $A$, let $|A|$ denote its cardinality. For a graph $G$, let $V(G)$ and $E(G)$ be its vertex set and edge set, respectively, and let $n:=|V(G)|$ be its order. The graph $G$ is complete if every pair of distinct vertices is adjacent. For a set $X\subseteq V(G)$, the notation $G-X$ denotes the subgraph induced by $V(G)\setminus X$, and $c(G-X)$ denotes its number of connected components. A cycle $C$ of $G$ is Hamiltonian if $V(C)=V(G)$, and $G$ is Hamiltonian if it contains a Hamilton cycle. For a real number $t>0$, the graph $G$ is $t$-tough if
\[
|X|\ge t\,c(G-X)
\]
for every $X\subseteq V(G)$ with $c(G-X)\ge2$; complete graphs are taken to be $t$-tough for every $t>0$. Chv\'atal conjectured that there exists an absolute constant $t_0>0$ such that every $t_0$-tough graph on at least three vertices is Hamiltonian \cite{Chvatal1973}. The conjecture remains open; see Bauer, Broersma, and Schmeichel for a survey \cite{BauerSurvey2006}. The examples of Bauer, Broersma, and Veldman show that such a constant, if it exists, is at least $9/4$ \cite{BauerBroersmaVeldman2000}.

For a vertex $v\in V(G)$, let $N_G(v)$ be its open neighbourhood, let $d_G(v):=|N_G(v)|$ be its degree, and let $\delta(G):=\min\{d_G(v):v\in V(G)\}$ be the minimum degree. Bauer, Broersma, van den Heuvel, and Veldman proved that a $t$-tough graph of order $n\ge3$ is Hamiltonian whenever $\delta(G)>n/(t+1)-1$ \cite{Bauer1995}.

For an integer $k\ge1$, let $P_k$ denote the path on $k$ vertices. If $F_1$ and $F_2$ are graphs, then $F_1\cup F_2$ denotes their vertex-disjoint union; for an integer $r\ge1$, the notation $rF_1$ denotes the disjoint union of $r$ copies of $F_1$. A graph is $F$-free if it has no induced subgraph isomorphic to $F$. A linear forest is a graph whose components are paths. Thus $2P_2\cup P_1$ is the disjoint union of two edges and one isolated vertex.

One approach to Chv\'atal's conjecture is to restrict the induced subgraphs of the graph. For $2P_2$-free graphs, the sufficient toughness bound was successively reduced from $25$ to $3$ and then to $2$ \cite{BroersmaPatelPyatkin2014,Shan2020,OtaSanka2022}. A parallel programme concerns forbidden linear forests on five vertices. Representative results cover $P_2\cup P_3$, $P_3\cup2P_1$, and $P_4\cup P_1$ \cite{Shan2021,GaoShan2022,Shan2026P4}; related refinements for $P_2\cup kP_1$ appear in \cite{ShiShan2022,XuLiZhou2024,Sanka2026}. Shan and Tanyel recently proved that every $11$-tough $(2P_2\cup P_1)$-free graph is Hamiltonian \cite{ShanTanyel2026}. Together with the preceding work, their theorem settles the existence of a sufficient toughness constant for every five-vertex linear forest other than $P_5$. Tanyel subsequently noted that it was unclear whether the bound $11$ is best possible and identified its improvement as an open direction \cite{Tanyel2026Dissertation}. The question addressed here is whether the explicit bound for $2P_2\cup P_1$ can be lowered.

Our main result is the following.

\begin{theorem}\label{thm:main}
Every $10$-tough $(2P_2\cup P_1)$-free graph on at least three vertices is Hamiltonian.
\end{theorem}

The proof of Theorem~\ref{thm:main} centres on an asymmetric analysis of the two components left by a putative small cut. Put $a:=n/11$. In the large-neighbourhood branch, define
\[
S:=\{v\in V(G):d_G(v)<2a\},\qquad
S_1:=\{v\in S:d_G(v)<a\},
\]
let $\ell:=|S_1|$, and let $G_2:=G-S$. For a graph $H$, write $\kappa(H)$ for its vertex-connectivity: $\kappa(H)=|V(H)|-1$ if $H$ is complete, and otherwise $\kappa(H)$ is the minimum size of a vertex set whose deletion disconnects $H$. The asymmetric analysis forces one component to have order in $(a,2a]$ and the other to have order greater than $6a$. A resulting partition into independent vertex classes, with every two vertices from different classes adjacent, yields
\[
\kappa(G_2)\ge \ell+\frac{n}{11}.
\]
This asymmetric small-cut lemma is combined with the compression of vertex-disjoint paths to pairwise vertex-disjoint added edges and with a cycle theorem for prescribed edges.

For an outline, the minimum-degree theorem handles $\delta(G)>a-1$. In the residual range, we split according to whether some edge $uv\in E(G)$ satisfies
\[
|N_G(u)\cup N_G(v)|\le4a.
\]
In the small-neighbourhood case, the sparse side is covered by vertex-disjoint paths with distinct endpoints on the complementary side. Replacing each path by an edge gives pairwise vertex-disjoint prescribed edges. In the large-neighbourhood case, the set of vertices of degree less than $2a$ is independent. The asymmetric cut argument supplies the connectivity bound above; a family of vertex-disjoint two-edge stars then inserts the vertices of degree less than $a$, and a cycle-extension lemma absorbs the remaining low-degree vertices.

Section~\ref{sec:prelim} gives the notation and auxiliary results. Section~\ref{sec:small} proves the small-neighbourhood proposition. Section~\ref{sec:large} develops the asymmetric large-neighbourhood argument, and Section~\ref{sec:proof} combines the two branches.

\section{PRELIMINARIES}\label{sec:prelim}

This section defines the notation used throughout the paper and records the auxiliary results invoked later.

\subsection{Sets, graphs, and induced subgraphs}

All graphs are finite and simple. For a graph $H$, let $V(H)$ and $E(H)$ denote its vertex set and edge set. The order of $H$ is $|V(H)|$, where $|A|$ denotes the cardinality of a finite set $A$. If $A$ and $B$ are sets, then $A\setminus B$ is the set of elements of $A$ not contained in $B$. For a real number $x$, the symbols $\lfloor x\rfloor$ and $\lceil x\rceil$ denote the greatest integer at most $x$ and the least integer at least $x$, respectively.

For a set $U\subseteq V(H)$, let $H[U]$ be the subgraph of $H$ induced by $U$, and let
\[
H-U:=H[V(H)\setminus U].
\]
For a vertex $x\in V(H)$, write $H-x$ for $H-\{x\}$. Let $c(H)$ denote the number of connected components of $H$. A component is \emph{trivial} if it has one vertex and \emph{nontrivial} otherwise. When a component is used as a set, it means its vertex set; thus, for example, $|D|$ denotes the order of a component $D$, and $H[D]$ denotes the subgraph induced by its vertices. A set $X\subseteq V(H)$ is a \emph{cutset} or \emph{vertex cut} if $c(H-X)\ge2$. Two disjoint vertex sets are \emph{anticomplete} if no edge has one endpoint in each set.

\subsection{Paths, cycles, and forbidden induced subgraphs}

For an integer $k\ge1$, let $P_k$ be the path on $k$ vertices, and let $K_k$ be the complete graph on $k$ vertices. For $k\ge1$, the graph $K_{1,k}$ is the star with one centre and $k$ leaves. A graph of order $m$ is complete if it is isomorphic to $K_m$. A vertex set is independent if no two distinct vertices in it are adjacent. A graph is bipartite if its vertex set can be partitioned into two independent sets. More generally, a graph is complete multipartite if its vertex set can be partitioned into independent sets, called multipartition classes, such that every two vertices in different classes are adjacent.

If $F_1$ and $F_2$ are graphs, then $F_1\cup F_2$ denotes their vertex-disjoint union. For an integer $r\ge1$, the notation $rF_1$ denotes the vertex-disjoint union of $r$ copies of $F_1$. A graph $H$ is $F$-free if no induced subgraph of $H$ is isomorphic to $F$. A linear forest is a graph whose components are paths. In particular, $2P_2\cup P_1$ consists of two vertex-disjoint edges and an isolated vertex, while $P_2\cup P_1$ consists of an edge and an isolated vertex.

For a path $P$, an endpoint is a vertex incident with at most one edge of $P$; all other vertices of $P$ are internal vertices. A cycle $C$ of $H$ is \emph{Hamiltonian} if $V(C)=V(H)$. The graph $H$ is Hamiltonian if it contains a Hamilton cycle. A path $P$ of $H$ is a Hamiltonian $(x,y)$-path if its endpoints are $x$ and $y$ and $V(P)=V(H)$. The graph $H$ is Hamiltonian-connected if every pair of distinct vertices is joined by a Hamiltonian path.

\subsection{Neighbourhoods, degree, connectivity, and independence}

For distinct vertices $x,y\in V(H)$, write $x\sim y$ if $x$ and $y$ are adjacent, and write $x\nsim y$ otherwise. For $x\in V(H)$, the open neighbourhood $N_H(x)$ is the set of vertices adjacent to $x$, and the closed neighbourhood is $N_H[x]:=N_H(x)\cup\{x\}$. The degree of $x$ is
\[
d_H(x):=|N_H(x)|.
\]
For $U\subseteq V(H)$, define the external neighbourhood
\[
N_H(U):=\left(\bigcup_{x\in U}N_H(x)\right)\setminus U.
\]
For two vertices $x,y\in V(H)$, the set $N_H(x)\cup N_H(y)$ is called their \emph{joint neighbourhood}.
The minimum degree is $\delta(H):=\min\{d_H(x):x\in V(H)\}$. When the ambient graph is clear, we write $N(x)$, $N(U)$, $d(x)$, or $\deg(x)$ in place of $N_H(x)$, $N_H(U)$, or $d_H(x)$.

\begin{definition}[Vertex-connectivity]
If $H$ is complete, define $\kappa(H):=|V(H)|-1$. If $H$ is not complete, define $\kappa(H)$ to be the minimum cardinality of a vertex cut of $H$.
\end{definition}

\begin{definition}[Independence number]
A set $I\subseteq V(H)$ is independent if no two distinct vertices of $I$ are adjacent. The independence number $\alpha(H)$ is the maximum cardinality of an independent set in $H$.
\end{definition}

\subsection{Toughness and the scattering number}

\begin{definition}[Toughness]
Let $t>0$ be a real number. A graph $H$ is $t$-tough if
\[
|X|\ge t\,c(H-X)
\]
for every $X\subseteq V(H)$ with $c(H-X)\ge2$. A complete graph is declared to be $t$-tough for every $t>0$.
\end{definition}

\begin{definition}[Scattering number]
If $H$ is not complete, its scattering number is
\[
s(H):=\max\{c(H-X)-|X|:X\subseteq V(H),\ c(H-X)\ge2\}.
\]
If $H$ is complete, set $s(H):=-\infty$.
\end{definition}

\subsection{Matchings, star-matchings, and matched path-covers}

A \emph{matching} in $H$ is a set of pairwise vertex-disjoint edges. A \emph{star-matching} is a union of pairwise vertex-disjoint stars. For $k\ge1$, a $K_{1,k}$-matching is a star-matching in which every component is isomorphic to $K_{1,k}$; the vertex of degree $k$ in each component is its centre.

Let $U,W\subseteq V(H)$ be disjoint. A $W$-matched path-cover of $U$ is a finite set $\mathcal{Q}$ of pairwise vertex-disjoint paths in $H$ such that:
\begin{enumerate}
\item the sets $V(P)\cap U$, for $P\in\mathcal{Q}$, partition $U$;
\item for every $P\in\mathcal{Q}$, the set $V(P)\setminus U$ consists exactly of the two distinct endpoints of $P$, both in $W$;
\item the $2|\mathcal{Q}|$ endpoints of the paths are pairwise distinct.
\end{enumerate}

\subsection{Hamiltonicity tools}

We next state four results used later.

\begin{lemma}[Minimum-degree toughness lemma]\label{lem:degree}
Let $t>0$ be a real number and let $G$ be a $t$-tough finite simple graph on $n\ge 3$ vertices, where $t$-tough means that for every vertex subset $X$ of $G$ with $c(G-X)\ge 2$ one has $|X|\ge t\cdot c(G-X)$. Assume $\delta(G)>n/(t+1)-1$. Then $G$ is Hamiltonian.
\end{lemma}

Lemma~\ref{lem:degree} is the theorem of Bauer, Broersma, van den Heuvel, and Veldman \cite{Bauer1995}. The same statement appears as Lemma~2.2 in \cite{Shan2021}, as Lemma~4 in \cite{ShiShan2022}, and as the minimum-degree lemma in \cite{ShanTanyel2026}.

\begin{lemma}[Matching-cycle lemma]\label{lem:matching-cycle}
Let $H$ be a finite simple graph with at least three vertices, and let $L\subseteq E(H)$ be a matching. If $\kappa(H)\ge |L|+\alpha(H)$, then $H$ has a Hamilton cycle containing every edge of $L$. Here a matching is a finite set of pairwise vertex-disjoint edges of $H$; $\kappa(H)$ equals $|V(H)|-1$ if $H$ is complete and otherwise equals the minimum cardinality of a vertex cut of $H$; $\alpha(H)$ is the independence number of $H$; and a Hamilton cycle containing every edge of $L$ is a cycle of $H$ that visits every vertex of $H$ exactly once and includes every edge of $L$ as a cycle edge.
\end{lemma}

Lemma~\ref{lem:matching-cycle} follows from the specified-edge cycle theorem of H\"aggkvist and Thomassen \cite{HaggkvistThomassen1982}.

\begin{lemma}[Cycle-extension lemma]\label{lem:cycle-ext}
Let $t>0$ be a real number and let $G$ be a $t$-tough finite simple graph on $n$ vertices, where $t$-tough means that for every vertex subset $X$ of $G$ with $c(G-X)\ge 2$ one has $|X|\ge t\cdot c(G-X)$. Let $C$ be a non-Hamiltonian cycle of $G$, and let $R$ be a connected subgraph of $G-V(C)$. If $|N_G(V(R))\cap V(C)|>n/(t+1)-1$, then there exists a cycle $C^*$ of $G$ such that $V(C)\subseteq V(C^*)$ and $V(C^*)\cap V(R)\neq\emptyset$.
\end{lemma}

Lemma~\ref{lem:cycle-ext} is Lemma~2.16 of \cite{Shan2026P4}; the same vertex-insertion statement is recorded in \cite{ShanTanyel2026}. A specialised singleton form appears as Lemma~2 of \cite{SankaShan2024}.

\begin{lemma}[Star-matching lemma]\label{lem:star}
Let $G$ be a finite simple graph that is not complete. Assume $G$ is $10$-tough in the following sense: for every vertex subset $X$ of $G$ with $c(G-X)\ge 2$ one has $|X|\ge 10\cdot c(G-X)$. Let $I$ be an independent set of vertices of $G$. Then $G$ has a $K_{1,2}$-matching whose set of centres is exactly $I$. That is: there exist pairwise vertex-disjoint copies of $K_{1,2}$ in $G$, one for each vertex of $I$, such that the centre of the copy associated to a vertex $x\in I$ is $x$ itself, and the two leaves of that copy are distinct vertices of $V(G)\setminus I$. In particular the $2|I|$ leaves are pairwise distinct, and no vertex of $I$ is a leaf of the matching.
\end{lemma}

Lemma~\ref{lem:star} is the $K_{1,2}$ case of the star-matching lemma in \cite{ShanTanyel2026}, which supplies a $K_{1,\lfloor t\rfloor}$-matching of centres $I$ in every $t$-tough noncomplete graph (here $t=10$).

\subsection{An elementary degree bound}

The next lemma is the residual minimum-degree consequence of $10$-toughness. It is used only to force the graph into a range where $a=n/11$ is large.

\begin{lemma}[Minimum degree under $10$-toughness]\label{lem:delta20}
Let $G$ be a finite simple graph that is not complete. Assume $G$ is $10$-tough in the following sense: for every vertex subset $X$ of $G$ with $c(G-X)\ge 2$ one has $|X|\ge 10\cdot c(G-X)$. Then $\delta(G)\ge 20$.
\end{lemma}

\begin{proof}
Let $v$ be a vertex of $G$ with $d_G(v)=\delta(G)$, and write $N[v]:=N_G(v)\cup\{v\}$. Because $G$ is not complete, the closed neighbourhood $N[v]$ is a proper subset of $V(G)$: otherwise every vertex other than $v$ is adjacent to $v$, so $\delta(G)=|V(G)|-1$ and $G$ is complete. Hence there exists a vertex $w\in V(G)\setminus N[v]$. In particular $w\neq v$ and $vw$ is not an edge of $G$.

Set $X:=N_G(v)$. In the induced subgraph $G-X$ the vertex $v$ has no remaining neighbour, so the singleton $\{v\}$ is a connected component of $G-X$. The vertex $w$ also belongs to $V(G-X)$, and $w\neq v$, so $c(G-X)\ge 2$. Ten-toughness therefore gives $|X|\ge 20$. Hence $\delta(G)=|X|\ge 20$.
\end{proof}

\subsection{The two neighbourhood propositions}

The proof of Theorem~\ref{thm:main} splits at the joint-neighbourhood threshold $4n/11$. The two sides of that split are packaged in the next two propositions. Throughout both statements, open neighbourhoods and degrees are taken in the ambient graph $G$ unless a subscript indicates otherwise.

\begin{proposition}[Small-neighbourhood Proposition]\label{prop:small}
Let $G$ be a finite simple graph on $n\ge 3$ vertices with no induced subgraph isomorphic to $2P_2 \cup P_1$. Assume that $G$ is $10$-tough in the following sense: for every vertex subset $X$ of $G$ with $c(G-X)\ge 2$ one has $|X|\ge 10\cdot c(G-X)$. Write $a:=n/11$. Suppose there exist adjacent vertices $u$ and $v$ of $G$ such that $|N_G(u)\cup N_G(v)|\le 4n/11$. Set $S:=(N_G(u)\cup N_G(v))\setminus\{u,v\}$. Then the following assertions hold.
\begin{enumerate}
\item The identity $S\cup\{u,v\}=N_G(u)\cup N_G(v)$ holds, and $|S\cup\{u,v\}|\le 4n/11$. The set $S$ is a cutset of $G$. The induced subgraph $D_1:=G[\{u,v\}]$ is a connected component of $G-S$, and $c(G-S)=2$.
\item Writing $D_2$ for the unique connected component of $G-S$ distinct from $D_1$, one has $|V(D_2)|\ge 7a$. Moreover $D_1$ is a nontrivial component, and $D_2$ is $(P_2 \cup P_1)$-free.
\item Let $S_1:=\{x\in S : |N_G(x)\cap V(D_2)|<2a\}$ and let $G_1:=G[S_1\cup\{u,v\}]$. Then $G_1$ is $(P_2 \cup P_1)$-free.
\item Let $S_2:=S\setminus S_1$ and let $G_2:=G[S_2\cup V(D_2)]$. Then $G_1$ admits a $V(G_2)$-matched path-cover $Q$ in $G$. If $G_1$ is complete then $|Q|=1$, while if $G_1$ is not complete then $|Q|=\max\{1,s(G_1)\}$.
\item Writing $\kappa(G_2)$ for the vertex-connectivity of $G_2$, one has $\kappa(G_2)\ge 2a$. Let $L$ be a finite set of edges with both endpoints in $V(G_2)$ and with $|L|\le a$. Write $G_2^*$ for the graph obtained from $G_2$ by adding every edge of $L$. Then $\kappa(G_2^*)\ge\kappa(G_2)\ge 2a$ and $\alpha(G_2^*)\le\alpha(G)\le a$. In particular $\kappa(G_2^*)\ge |L|+\alpha(G_2^*)$.
\end{enumerate}
\end{proposition}

\section{THE SMALL-NEIGHBOURHOOD CASE}\label{sec:small}

Assume that an edge $uv$ satisfies $|N_G(u)\cup N_G(v)|\le4n/11$, and write $a:=n/11$. We prove the three ingredients of Proposition~\ref{prop:small}: the two-component decomposition, a matched path-cover of the sparse side, and the connectivity bound needed after the paths are compressed. The balanced path-cover branch is handled separately using the minimum-degree consequence $\delta(G)\ge20$.

\subsection{Two-component structure}

The following proposition establishes items (1)--(3) of Proposition~\ref{prop:small}. Its numerical core is that the large component has order at least seven times the natural scale, whereas the three vertices of a hypothetical forbidden configuration have fewer than six times that scale in the component as a whole; the remaining vertices then exceed the permitted independence number.

\begin{proposition}[Small-edge neighbourhood structure]\label{prop:small-struct}
Let $G$ be a finite simple graph on $n\ge 3$ vertices with no induced subgraph isomorphic to $2P_2 \cup P_1$. Assume that $G$ is $10$-tough in the following sense: for every vertex subset $X$ of $G$ with $c(G-X)\ge 2$ one has $|X|\ge 10\cdot c(G-X)$. Write $a:=n/11$. Suppose there exist adjacent vertices $u$ and $v$ of $G$ such that $|N_G(u)\cup N_G(v)|\le 4n/11$. Set $S:=(N_G(u)\cup N_G(v))\setminus\{u,v\}$. Then the following three assertions hold.
\begin{enumerate}
\item The set $\{u,v\}$ is contained in $N_G(u)\cup N_G(v)$, the identity $S\cup\{u,v\}=N_G(u)\cup N_G(v)$ holds, $|S\cup\{u,v\}|\le 4n/11$, the set $S$ is a cutset of $G$, the induced subgraph $D_1:=G[\{u,v\}]$ is a connected component of $G-S$, and $c(G-S)=2$.
\item Writing $D_2$ for the unique connected component of $G-S$ distinct from $D_1$, one has $|V(D_2)|\ge 7a$. Moreover $D_1$ is a nontrivial component, and $D_2$ is $(P_2 \cup P_1)$-free.
\item Let $S_1:=\{x\in S : |N_G(x)\cap V(D_2)|<2a\}$ and let $G_1:=G[S_1\cup\{u,v\}]$. Then $G_1$ is $(P_2 \cup P_1)$-free.
\end{enumerate}
\end{proposition}

\begin{proof}
We work throughout with the finite simple graph $G$ in the statement. Put $a:=n/11$, choose adjacent $u,v$ with $|N_G(u)\cup N_G(v)|\le4n/11$, and set
\[
S:=(N_G(u)\cup N_G(v))\setminus\{u,v\}.
\]
Here all neighbourhoods are open.

Because $uv$ is an edge, $v\in N_G(u)$ and $u\in N_G(v)$. Hence
\[
S\cup\{u,v\}=N_G(u)\cup N_G(v),\qquad |S|\le4n/11-2.
\]
Also $|N_G(u)\cup N_G(v)|\ge2$, so $n\ge6$. The graph $G$ is not complete: otherwise $N_G(u)\cup N_G(v)=V(G)$ and $n\le4n/11$, impossible.

Since $G$ is noncomplete, choose a maximum independent set $I$. Then $|I|\ge2$, and for $X:=V(G)\setminus I$ one has $c(G-X)=|I|$. Toughness gives
\[
n-|I|=|X|\ge10|I|,
\]
so $\alpha(G)=|I|\le a$.

We next prove the component statement. The edge $uv$ makes $D_1:=G[\{u,v\}]$ connected. Neither $u$ nor $v$ has a neighbour outside $S\cup\{u,v\}$, because all their neighbours lie in $N_G(u)\cup N_G(v)=S\cup\{u,v\}$. Thus $D_1$ is a component of $G-S$. The set $V(G)\setminus(S\cup\{u,v\})$ is nonempty because $|S\cup\{u,v\}|\le4n/11<n$, so $S$ is a cutset.

The component $D_1$ is nontrivial. There are at most two nontrivial components of $G-S$: if $D',D'',D'''$ were three, choose an edge from each of the first two and a vertex from the third. Distinct components are anticomplete, so these five vertices induce $2P_2\cup P_1$, a contradiction.

Suppose that a nontrivial component $D^*$ distinct from $D_1$ exists. A third component $D^{**}$ would again yield an induced $2P_2\cup P_1$ by taking the edge $uv$, an edge in $D^*$, and a vertex of $D^{**}$. Thus $G-S$ has exactly two components.

It remains to rule out the case in which all components other than $D_1$ are trivial. Then
\[
c(G-S)=1+(n-|S|-2)=n-|S|-1\ge7n/11+1.
\]
Toughness would imply
\[
|S|\ge10c(G-S)\ge70n/11+10,
\]
whereas $|S|\le4n/11-2$, forcing $-6n\ge12$, impossible. Hence the second component $D_2$ is nontrivial and is the unique component other than $D_1$.

The order bound follows immediately:
\[
\begin{aligned}
|V(D_2)|&=n-|S\cup\{u,v\}|\\
&=n-|N_G(u)\cup N_G(v)|\\
&\ge n-4n/11=7n/11=7a.
\end{aligned}
\]
If $D_2$ contained an induced $P_2\cup P_1$ on vertices $x,y,z$, then the edge $uv$ in $D_1$ together with $x,y,z$ would induce $2P_2\cup P_1$, since distinct components of $G-S$ are anticomplete. Thus $D_2$ is $(P_2\cup P_1)$-free.

Finally, set
\[
S_1:=\{x\in S:|N_G(x)\cap V(D_2)|<2a\},\qquad
G_1:=G[S_1\cup\{u,v\}].
\]
If $|V(G_1)|\le2$, the conclusion is immediate. Otherwise suppose that $G_1$ contains an induced $P_2\cup P_1$ on vertices $x,y,z$. These vertices lie outside $D_2$. Let $N(\{x,y,z\})$ denote their open neighbourhood in $G$. Neither $u$ nor $v$ has a neighbour in $D_2$, and each vertex of $S_1$ has fewer than $2a$ neighbours in $D_2$, so
\[
|N(\{x,y,z\})\cap V(D_2)|
\le\sum_{w\in\{x,y,z\}\cap S_1}|N_G(w)\cap V(D_2)|
<6a.
\]
Let $U:=V(D_2)\setminus N(\{x,y,z\})$. If $U$ contained an edge $pq$, then the edge in the induced $P_2\cup P_1$ on $\{x,y,z\}$, the edge $pq$, and the remaining isolated vertex would induce $2P_2\cup P_1$ in $G$. Hence $U$ is independent in $D_2$, and
\[
|U|\le\alpha(D_2)\le\alpha(G)\le a.
\]
On the other hand,
\[
|U|=|V(D_2)|-|N(\{x,y,z\})\cap V(D_2)|
>7a-6a=a,
\]
a contradiction. Therefore $G_1$ is $(P_2\cup P_1)$-free.
\end{proof}

\subsection{Matched path-cover}

We next prove the path-cover assertion. The negative-scattering case is handled by Jung's Hamiltonian-connectedness theorem for $P_4$-free graphs, while the balanced case uses the degree bound $\delta(G)\ge20$ and an explicit toughness argument.

\begin{proposition}[Matched path-cover with distinct outside ends]\label{prop:pathcover}
Let $G$ be a finite simple graph on $n\ge 3$ vertices with no induced subgraph isomorphic to $2P_2 \cup P_1$. Assume that $G$ is $10$-tough. Write $a:=n/11$. Suppose there exist adjacent vertices $u$ and $v$ of $G$ such that $|N_G(u)\cup N_G(v)|\le 4n/11$. Let $S$, $D_1$, $D_2$, $S_1$, and $G_1$ be as in Proposition~\ref{prop:small-struct}. Let $S_2:=S\setminus S_1$ and let $G_{\mathrm{right}}:=G[S_2\cup V(D_2)]$. If $G_1$ is not complete, write $s(G_1)$ for the scattering number of $G_1$. Then $G_1$ admits a $V(G_{\mathrm{right}})$-matched path-cover $Q$ in $G$. If $G_1$ is complete then $|Q|=1$, while if $G_1$ is not complete then $|Q|=\max\{1,s(G_1)\}$.
\end{proposition}

\begin{proof}
We use the notation of the statement. Proposition~\ref{prop:small-struct} gives that $G$ is noncomplete, $c(G-S)=2$, $|V(D_2)|\ge7a$, that $D_2$ and $G_1$ are $(P_2\cup P_1)$-free, that $\{u,v\}\subseteq V(G_1)$, and that $V(G)$ is the disjoint union of $V(G_1)$ and $V(G_{\mathrm{right}})$. In particular, $|V(G_1)|\ge2$.

We first record the degree and outside-size bounds used below. Let $x$ be a minimum-degree vertex of $G$ and put $X=N_G(x)$. Since $G$ is noncomplete, $G-X$ has at least two components, one containing $x$ and one containing a vertex outside $N_G[x]$ or a component of $(G-x)-X$. Thus 10-toughness gives $\delta(G)=|X|\ge20$. Moreover $n\ge6$ and $|V(G_{\mathrm{right}})|\ge|V(D_2)|\ge7a>3$, so $|V(G_{\mathrm{right}})|\ge4$.

The graph $G_1$ is complete multipartite. To see this, observe that nonadjacency is transitive: if $p\not\sim q$ and $q\not\sim r$ but $p\sim r$, then $q$ is isolated from the edge $pr$, producing an induced $P_2\cup P_1$. Hence the nonadjacency classes are independent and vertices from different classes are adjacent.

There is a matching of size two in the bipartite graph of cross-edges between $V(G_1)$ and $V(G_{\mathrm{right}})$. If there were no cross-edge, $G$ would be disconnected, which is impossible for a 10-tough graph. If all cross-edges shared one endpoint $r$, then after deleting $r$ the two nonempty sides would remain disconnected: $|V(G_1)|\ge2$ and $|V(G_{\mathrm{right}})|\ge4$. This would give $c(G-\{r\})\ge2$ and $1\ge20$, a contradiction. Fix distinct $x,y\in V(G_1)$ and distinct $z,w\in V(G_{\mathrm{right}})$ with $xz,yw\in E(G)$.

\medskip
\noindent\textbf{Step 1: complete or negatively scattering.}
If $G_1$ is complete, it is Hamiltonian-connected. A Hamiltonian $(x,y)$-path $P$ therefore gives the required one-path cover $z\,x\,P\,y\,w$.

Assume that $G_1$ is not complete and $s(G_1)\le-1$. Since $G_1$ is $(P_2\cup P_1)$-free, it is $P_4$-free: an induced path $P_4=p_1p_2p_3p_4$ would contain an induced subgraph on $\{p_1,p_2,p_4\}$ isomorphic to $P_2\cup P_1$. Moreover $s(G_1)<0$. Jung's theorem therefore implies that $G_1$ is Hamiltonian-connected \cite{Jung1978}. A Hamiltonian $(x,y)$-path $P$ again yields the single matched path $z\,x\,P\,y\,w$. Thus $|Q|=1=\max\{1,s(G_1)\}$.

\medskip
\noindent\textbf{Step 2: nonnegative scattering.}
Assume that $G_1$ is noncomplete and $s(G_1)\ge0$. Let $T$ be a minimum cutset of $G_1$, put $B:=V(G_1)\setminus T$, $k:=|B|$, and $n_1:=|V(G_1)|$.

We record the cut facts explicitly. Every component of $G_1-T$ is trivial, so $B$ is independent and $c(G_1-T)=k$. Every vertex of $T$ is adjacent to every vertex of $B$, so the complete bipartite graph with parts $(T,B)$ is a spanning subgraph of $G_1$. In a complete multipartite graph, $\kappa(G_1)=\delta(G_1)$ and $\delta(G_1)\ge n_1-\alpha(G_1)$. The set $T$ is nonempty because $uv$ is an edge. A largest multipartition class $I$ has size $\alpha(G_1)$ and its complement is a cutset, so
\[
|T|=\kappa(G_1)=n_1-\alpha(G_1),\qquad k=\alpha(G_1).
\]
For every cutset $X$, $G_1-X$ is edgeless and has at most $\alpha(G_1)$ vertices, hence
\[
s(G_1)=2\alpha(G_1)-n_1=k-|T|.
\]
Thus $|T|\le k$.

By Lemma~\ref{lem:star}, the independent set $B$ is the centre set of a $K_{1,2}$-matching $M$ in $G$. Its $2k$ leaves are pairwise distinct, none lies in $B$, and every leaf lies in $T\cup V(G_{\mathrm{right}})$ because $V(G)=T\cup B\cup V(G_{\mathrm{right}})$.

\medskip
\noindent\textbf{Step 3: the unbalanced branch $|T|<k$.}
Here $s(G_1)=k-|T|\ge1$. Let $r$ be the number of centres of $M$ whose two leaves lie in $T$. Then $2r\le|T|$, so $r\le\lfloor|T|/2\rfloor$. Let $U\subseteq B$ be the centres with at least one leaf in $T$. Since leaves are distinct, $|U|\le|T|$. Extend $U$ to $U^*\subseteq B$ with $|U^*|=|T|+1$.

At most $\lfloor|T|/2\rfloor$ centres of $U^*$ have both leaves in $T$. Consequently at least
\[
|T|+1-\left\lfloor\frac{|T|}{2}\right\rfloor\ge2
\]
centres $x,y\in U^*$ have an $M$-leaf in $V(G_{\mathrm{right}}).$ Choose distinct such leaves $z,w$. Enumerate $T=\{t_1,\ldots,t_m\}$ with $m=|T|$ and $U^*\setminus\{x,y\}=\{b_1,\ldots,b_{m-1}\}$. The alternating walk
\[
x\,t_1\,b_1\,t_2\,b_2\,\cdots\,t_{m-1}\,b_{m-1}\,t_m\,y
\]
is a Hamiltonian $(x,y)$-path of the complete bipartite spanning subgraph on $T\cup U^*$. Joining it to $z$ and $w$ gives one matched path. Every centre in $B\setminus U^*$ has both leaves in $V(G_{\mathrm{right}})$, so its two-edge star is another matched path. The resulting collection has
\[
1+(k-|U^*|)=k-|T|=s(G_1)
\]
paths, covers $T\cup B=V(G_1)$, and has pairwise distinct outside endpoints.

\medskip
\noindent\textbf{Step 4: the balanced branch $|T|=k$.}
Now $n_1=2k$ and $s(G_1)=0$. If $n_1\le19$, then for each $q\in V(G_1)$,
\[
d_{G_{\mathrm{right}}}(q)\ge\delta(G)-(n_1-1)\ge21-n_1\ge2.
\]
Choose $x\in T$ and $y\in B$, choose a neighbour $z$ of $x$ in $G_{\mathrm{right}}$, and choose a neighbour $w\ne z$ of $y$ in $G_{\mathrm{right}}$. The balanced complete bipartite graph on $(T,B)$ has a Hamiltonian $(x,y)$-path, so $z\,x\,P\,y\,w$ is one matched path covering $G_1$. Since $n_1$ is even, this case is exactly $n_1\le18$.

Assume now $n_1\ge20$, so $k\ge10$. At most $\lfloor k/2\rfloor$ centres have both leaves in $T$, so at least $\lceil k/2\rceil\ge5$ centres have an outside leaf. Choose two such centres $x,y\in B$ with distinct outside leaves $z,w$.

If $T$ contains an edge $t_1t_2$, use the complete bipartite edges together with $t_1t_2$ to form an alternating Hamiltonian $(x,y)$-path that traverses the detour $t_1t_2$ once. This gives a single matched path covering $V(G_1)$.

If $T$ is independent, then $G_1$ is the complete bipartite graph with parts $(T,B)$. We claim that some $x^*\in T$ has a neighbour $z^*\in V(G_{\mathrm{right}})\setminus\{z\}$. Suppose not. Then
\[
N_G(T)\cap V(G_{\mathrm{right}})\subseteq\{z\}.
\]
Delete $X:=B\cup\{z\}$. Every vertex of $T$ is isolated in $G-X$: its neighbours in $G_1$ all lie in $B$, and by assumption it has no neighbour in $G_{\mathrm{right}}\setminus\{z\}$. Since $|V(G_{\mathrm{right}})|\ge4$, the set $V(G_{\mathrm{right}})\setminus\{z\}$ is nonempty. Hence
\[
c(G-X)\ge k+1,\qquad |X|=k+1.
\]
Ten-toughness would give $k+1=|X|\ge10c(G-X)\ge10(k+1)$, a contradiction. The claimed edge $x^*z^*$ therefore exists. The complete bipartite graph $G_1$ has a Hamiltonian $(x,x^*)$-path $P$, and $z\,x\,P\,x^*\,z^*$ is the required one-path cover. This exhausts all cases and proves the proposition.
\end{proof}

\subsection{Connectivity after adding the endpoint matching}

\begin{proposition}[Connectivity of the large side]\label{prop:kappa-small}
Let $G$ be a finite simple graph on $n\ge 3$ vertices with no induced subgraph isomorphic to $2P_2 \cup P_1$. Assume that $G$ is $10$-tough. Write $a:=n/11$. Suppose there exist adjacent vertices $u$ and $v$ of $G$ such that $|N_G(u)\cup N_G(v)|\le 4n/11$. Set $S:=(N_G(u)\cup N_G(v))\setminus\{u,v\}$. Let $D_1:=G[\{u,v\}]$ and let $D_2$ be the unique connected component of $G-S$ distinct from $D_1$, as furnished by Proposition~\ref{prop:small-struct}. Let $S_1:=\{x\in S : |N_G(x)\cap V(D_2)|<2a\}$, write $S_2:=S\setminus S_1$, and write $G_2:=G[S_2\cup V(D_2)]$. Then $\kappa(G_2)\ge 2a$. Moreover, if $L$ is a finite set of edges with both endpoints in $V(G_2)$ and $|L|\le a$, and if $G_2^*$ is obtained from $G_2$ by adding every edge of $L$, then $\kappa(G_2^*)\ge\kappa(G_2)\ge 2a$ and $\alpha(G_2^*)\le\alpha(G)\le a$. In particular $\kappa(G_2^*)\ge |L|+\alpha(G_2^*)$.
\end{proposition}

\begin{proof}
By Proposition~\ref{prop:small-struct}, $S$ is a cutset, $c(G-S)=2$, $|V(D_2)|\ge7a$, $D_2$ is $(P_2\cup P_1)$-free, and $G_1:=G[S_1\cup\{u,v\}]$ is $(P_2\cup P_1)$-free. The sets $V(G_1)$ and $V(G_2)$ are disjoint and partition $V(G)$. Since $G$ is noncomplete, the minimum-degree lemma gives $\delta(G)\ge20$, hence $n\ge21$ and $a\ge21/11$. As above, $\alpha(G)\le a$ and $|V(G_2)|\ge|V(D_2)|\ge7a$.

If $G_2$ is complete, then
\[
\kappa(G_2)=|V(G_2)|-1\ge7a-1\ge2a,
\]
because $5a\ge1$. Suppose that $G_2$ is not complete, and let $W$ be a minimum vertex cut. Assume for a contradiction that $|W|<2a$.

The set $V(D_2)\setminus W$ has order
\[
|V(D_2)\setminus W|\ge|V(D_2)|-|W|>7a-2a=5a>a\ge\alpha(D_2).
\]
Thus it contains an edge $yz$. Let $Q_1$ be the component of $G_2-W$ containing $y,z$. Every vertex of $V(D_2)\setminus W$ lies in $Q_1$: otherwise a vertex $r$ in another component would be nonadjacent to both $y$ and $z$, and $\{y,z,r\}$ would induce $P_2\cup P_1$ in $D_2$.

Since $W$ is a vertex cut, $G_2-W$ has another component $Q_2$. By the preceding paragraph, $Q_2\subseteq S_2$. Choose $x\in Q_2$. If $x$ had a neighbour in $V(D_2)\setminus W$, that edge would join $Q_2$ to $Q_1$, impossible. Hence
\[
N_G(x)\cap V(D_2)\subseteq W.
\]
But $x\in S_2$ means $|N_G(x)\cap V(D_2)|\ge2a$, so $|W|\ge2a$, a contradiction. Therefore $\kappa(G_2)\ge2a$.

Let $L$ be as in the statement and let $G_2^*$ be obtained by adding its edges. Adding edges cannot create a new vertex cut or decrease the order of an existing cut, so $\kappa(G_2^*)\ge\kappa(G_2)\ge2a$. Every independent set of $G_2^*$ is independent in $G$, so $\alpha(G_2^*)\le\alpha(G)\le a$. Since $|L|\le a$,
\[
\kappa(G_2^*)\ge2a\ge|L|+a\ge|L|+\alpha(G_2^*).
\]
This proves the proposition.
\end{proof}

This completes the small-neighbourhood case: Proposition~\ref{prop:small-struct} supplies the two-component cut and the $(P_2 \cup P_1)$-free sides, Proposition~\ref{prop:pathcover} supplies the matched path-cover, and Proposition~\ref{prop:kappa-small} supplies $\kappa(G_2)\ge 2a$ together with the compression after adding at most $a$ edges. Proposition~\ref{prop:small} follows.

\section{THE LARGE-NEIGHBOURHOOD CASE}\label{sec:large}

This section derives the large-side connectivity bound from an asymmetric component analysis. After deleting the low-degree vertices and a hypothetical small vertex cut, exactly two nontrivial components remain. If both are comparatively large, a vertex of very low degree leaves an edge in each component outside its neighbourhood, producing the forbidden induced subgraph. Otherwise one component lies between one and two times the natural scale and the other exceeds six times that scale. The forbidden induced subgraph condition then forces a common multipartition class in the smaller component, contradicting the lower bound on the joint neighbourhood of an edge.

\begin{proposition}[Large-neighbourhood proposition]\label{prop:large}
Let $G$ be a finite simple graph on $n$ vertices with no induced subgraph isomorphic to $2P_2 \cup P_1$. Write $a:=n/11$, and assume $\alpha(G)\le a$. Assume that every edge $uv$ of $G$ satisfies $|N_G(u)\cup N_G(v)|>4a$. Define $S:=\{v\in V(G):\deg(v)<2a\}$, $S_1:=\{x\in S:\deg(x)<a\}$, $\ell:=|S_1|$, and $G_2:=G-S$. Assume $S_1$ is nonempty. Then:
\begin{enumerate}
\item $S$ is independent and $|S|\le a$, and hence $|V(G_2)|\ge10a$;
\item $G_2$ has no vertex cut of order less than $\ell+a$;
\item if $G_2$ is not complete, then $\kappa(G_2)\ge\ell+a$, while if $G_2$ is complete, then $\kappa(G_2)=|V(G_2)|-1\ge10a-1$.
\end{enumerate}
\end{proposition}

Neighbourhoods and degrees in this section are taken in $G$. We first isolate the asymmetric configuration that rules out the only remaining small-cut geometry.

\subsection{Asymmetric components}

\begin{lemma}[Asymmetric neighbourhood configuration]\label{lem:asymm}
Let $G$ be a finite simple graph, and let $a>0$ be a real number. Assume $G$ has no induced subgraph isomorphic to $2P_2 \cup P_1$, and assume $\alpha(G)\le a$. Assume that every edge $uv$ of $G$ satisfies $|N(u)\cup N(v)|>4a$. Let $S$ and $W$ be subsets of $V(G)$. Let $S_1$ be a nonempty subset of $S$ such that every vertex of $S_1$ has degree strictly less than $a$ in $G$, and write $\ell:=|S_1|$. Assume that $S$ is an independent set, that $|S|\le a$, and that every vertex of $S$ has degree strictly less than $2a$ in $G$. Assume $|W|<\ell+a$. Let $A$ and $B$ be distinct connected components of $G-(S\cup W)$. Assume that $G[A]$ is complete multipartite, that $|A|>a$, and that $|B|>6a$. Then these hypotheses are inconsistent.
\end{lemma}

\begin{proof}
Possibly replacing $W$ with $W\setminus S$, we may assume that $W$ is disjoint from $S$. The sets $A$ and $B$ are then disjoint from $S\cup W$, and distinct components of $G-(S\cup W)$ are anticomplete.

\medskip
\noindent\textbf{Step 1.} The induced subgraph $H:=G[A\cup S]$ is $P_4$-free. Suppose that $P=p_1p_2p_3p_4$ is an induced path in $H$, and put $X:=V(P)\cap S$. Since an induced $P_4$ has independence number two and $S$ is independent, $|X|\le2$. Set $B_0:=B\setminus N_G(X)$. Every vertex of $S$ has degree strictly less than $2a$, and hence
\[
|N_G(X)|\le\sum_{x\in X}\deg(x)<|X|\,2a\le4a.
\]
Together with $|B|>6a$, this gives $|B_0|>2a$. Every vertex of $B_0$ is anticomplete to $P$: it has no neighbour in $X$ by construction and no neighbour in $V(P)\cap A$ because $A$ and $B$ are distinct components. If $B_0$ were independent, then $|B_0|>a\ge\alpha(G)$, a contradiction. Choose an edge $yz$ in $B_0$. The five vertices $\{y,z,p_1,p_2,p_4\}$ induce $2P_2\cup P_1$, contradicting the forbidden-subgraph hypothesis. Thus $H$ is $P_4$-free.

\medskip
\noindent\textbf{Step 2.} Fix $x\in S_1$. The set $L_x:=A\setminus N_G(x)$ is nonempty because
\[
|L_x|\ge|A|-\deg(x)>a-a=0.
\]
It is independent. Otherwise an edge $uv$ in $L_x$ together with an edge in $B\setminus N_G(x)$ and the vertex $x$ would induce $2P_2\cup P_1$; the required edge in $B\setminus N_G(x)$ exists because
\[
|B\setminus N_G(x)|\ge |B|-\deg(x)>6a-a=5a\ge\alpha(G).
\]
Since $G[A]$ is complete multipartite, $L_x$ lies in a unique multipartition class $P_x$ of $G[A]$. Writing $C_x:=A\setminus P_x$, every vertex of $C_x$ is adjacent to $x$.

The classes $P_x$ coincide for all $x\in S_1$. If $x,y\in S_1$ had $P_x\ne P_y$, choose $p\in L_x\subseteq P_x$ and $q\in L_y\subseteq P_y$. Then $p\in C_y$ and $q\in C_x$, so $yp$ and $xq$ are edges; also $pq$ is an edge because the two vertices lie in distinct multipartition classes. The set $S$ is independent, and $xp$ and $yq$ are nonedges by the definitions of $L_x$ and $L_y$. Therefore $x-q-p-y$ is an induced $P_4$ in $H$, a contradiction. Let $P$ be this common class and put $C:=A\setminus P$. Every vertex of $S_1$ is complete to $C$, so $|C|<a$.

\medskip
\noindent\textbf{Step 3.} There is no edge between $S_1$ and $P$. Suppose $xy\in E(G)$ with $x\in S_1$ and $y\in P$. Every neighbour of $y$ outside $S\cup W$ lies in the component $A$; inside $A$, the vertex $y$ has no neighbour in its own part $P$. Hence
\[
N_G(y)\subseteq C\cup S\cup W\subseteq N_G(x)\cup S\cup W.
\]
Consequently
\[
\begin{aligned}
|N_G(x)\cup N_G(y)|
&\le \deg(x)+|S|+|W|\\
&< a+|S|+(\ell+a)\\
&=2a+|S|+\ell\le4a,
\end{aligned}
\]
contradicting the standing edge-neighbourhood hypothesis. Thus $P\cup S_1$ is independent, and
\[
|P|+\ell\le\alpha(G)\le a.
\]
Together with $|C|<a$ and $|S|\le a$, this yields
\[
|A|+|S|+\ell=|C|+|P|+|S|+\ell<3a.
\]

\medskip
\noindent\textbf{Step 4.} The set $A$ is not independent, since $|A|>a\ge\alpha(G)$. Choose an edge $uv$ in $G[A]$. Every neighbour of $u$ or $v$ outside $S\cup W$ lies in the component $A$, so
\[
N_G(u)\cup N_G(v)\subseteq A\cup S\cup W.
\]
Therefore
\[
|N_G(u)\cup N_G(v)|
 <|A|+|S|+\ell+a
 <3a+a=4a,
\]
again contradicting the edge-neighbourhood hypothesis. This contradiction proves the lemma.
\end{proof}

\subsection{\texorpdfstring{Small cuts of $G_2$}{Small cuts of G2}}

\begin{proposition}[No small cut]\label{prop:no-small-cut}
Let $G$ be a finite simple graph on $n$ vertices with no induced subgraph isomorphic to $2P_2 \cup P_1$. Write $a:=n/11$, and assume $\alpha(G)\le a$. Assume that every edge $uv$ of $G$ satisfies $|N(u)\cup N(v)|>4a$. Define $S:=\{v\in V(G):\deg(v)<2a\}$, $S_1:=\{x\in S:\deg(x)<a\}$, $\ell:=|S_1|$, and $G_2:=G-S$. Assume $S_1$ is nonempty. Then $G_2$ has no vertex cut $W$ with $|W|<\ell+a$.
\end{proposition}

\begin{proof}
The identity $n=11a$ follows from $a:=n/11$.

\medskip
\noindent\textbf{Step 1.} If $xy$ were an edge with both ends in $S$, then the standing hypothesis would give $|N(x)\cup N(y)|>4a$, whereas
\[
|N(x)\cup N(y)|\le\deg(x)+\deg(y)<4a,
\]
a contradiction. Thus $S$ is independent and $|S|\le\alpha(G)\le a$, so $\ell\le a$. Suppose for a contradiction that $W$ is a vertex cut of $G_2$ with $|W|<\ell+a$. Since $W\subseteq V(G_2)=V(G)\setminus S$, it is disjoint from $S$, and
\[
|S\cup W|=|S|+|W|<|S|+\ell+a\le3a.
\]

\medskip
\noindent\textbf{Step 2.} The graph $G_2-W=G-(S\cup W)$ has no trivial component. If $\{z\}$ were such a component, then $z\notin S$ and $\deg(z)\ge2a$, while $N(z)\subseteq S\cup W$. If $z$ had a neighbour $x\in S_1$, then
\[
N(x)\cup N(z)\subseteq N(x)\cup S\cup W
\]
would imply
\[
|N(x)\cup N(z)|<a+|S|+|W|<4a,
\]
contradicting that $xz$ is an edge. Hence $z$ has no neighbour in $S_1$, and
\[
\deg(z)\le |S\setminus S_1|+|W|=|S|-\ell+|W|<2a,
\]
contradicting $\deg(z)\ge2a$.

\medskip
\noindent\textbf{Step 3.} Hence $G_2-W$ has at least two nontrivial components. It has at most two: three nontrivial components would contain two disjoint edges and a vertex from the third component, inducing $2P_2\cup P_1$. Let $Q_1,Q_2$ be the two components. Each is $(P_2\cup P_1)$-free, because an induced $P_2\cup P_1$ in either component, together with an edge of the other component, would contradict the forbidden-subgraph hypothesis. Consequently each $Q_i$ is complete multipartite.

\medskip
\noindent\textbf{Step 4.} Each $Q_i$ contains an edge $uv$. Since all neighbours of $u$ and $v$ outside $S\cup W$ lie in $Q_i$,
\[
|N(u)\cup N(v)|\le |Q_i|+|S\cup W|,
\]
so $|Q_i|>4a-|S\cup W|>a$. Moreover
\[
|Q_1|+|Q_2|=n-|S\cup W|>8a.
\]

\medskip
\noindent\textbf{Step 5.} Suppose $|Q_1|>2a$ and $|Q_2|>2a$. Fix $x\in S_1$. If $\deg(x)\le a$, then $|Q_i\setminus N(x)|>a\ge\alpha(G)$ for $i=1,2$, so each leftover set contains an edge. The two edges together with $x$ induce $2P_2\cup P_1$, contradiction. Thus $\deg(x)>a$, contradicting $x\in S_1$.

\medskip
\noindent\textbf{Step 6.} We are left with the case that one of the two components has order at most $2a$. Relabel them as $A,B$ with $|A|\le|B|$. The bounds above give $a<|A|\le2a$ and
\[
|B|=|Q_1|+|Q_2|-|A|>8a-2a=6a.
\]
The sets $A,B,S,W,S_1$ satisfy all hypotheses of Lemma~\ref{lem:asymm}, which is impossible. This contradiction proves that no such cut $W$ exists.
\end{proof}

\subsection{\texorpdfstring{Connectivity of $G_2$}{Connectivity of G2}}

\begin{corollary}[Connectivity of the large-neighbourhood graph]\label{cor:kappa-large}
Let $G$ be a finite simple graph on $n$ vertices with no induced subgraph isomorphic to $2P_2 \cup P_1$. Write $a:=n/11$, and assume $\alpha(G)\le a$. Assume that every edge $uv$ of $G$ satisfies $|N(u)\cup N(v)|>4a$. Define $S:=\{v\in V(G):\deg(v)<2a\}$, $S_1:=\{x\in S:\deg(x)<a\}$, $\ell:=|S_1|$, and $G_2:=G-S$. Assume $S_1$ is nonempty. Then $S$ is independent and $|S|\le a$, hence $|V(G_2)|\ge 10a$. If $G_2$ is not complete, then $\kappa(G_2)\ge\ell+a$. If $G_2$ is complete, then $\kappa(G_2)=|V(G_2)|-1\ge 10a-1$.
\end{corollary}

\begin{proof}
Independence of $S$ is as in Proposition~\ref{prop:no-small-cut}, so $|S|\le a$ and $|V(G_2)|\ge10a$. By Proposition~\ref{prop:no-small-cut}, $G_2$ has no vertex cut of order less than $\ell+a$. If $G_2$ is not complete, its connectivity is the minimum cut order and hence $\kappa(G_2)\ge\ell+a$. If $G_2$ is complete, then by definition $\kappa(G_2)=|V(G_2)|-1\ge10a-1$.
\end{proof}

Proposition~\ref{prop:large} is the conjunction of Corollary~\ref{cor:kappa-large} and Proposition~\ref{prop:no-small-cut}.

\section{PROOF OF THE MAIN THEOREM}\label{sec:proof}

In this section we prove Theorem~\ref{thm:main} from Propositions~\ref{prop:small} and \ref{prop:large}, together with Lemmas~\ref{lem:degree}, \ref{lem:matching-cycle}, \ref{lem:cycle-ext}, \ref{lem:star}, and \ref{lem:delta20}.

\begin{proof}[Proof of Theorem~\ref{thm:main}]
Let $G$ be a finite simple graph on $n\ge 3$ vertices with no induced subgraph isomorphic to $2P_2 \cup P_1$, and assume that $G$ is $10$-tough. Write $a:=n/11$. Open neighbourhoods and degrees are in $G$ unless a subscript indicates otherwise.

If $\delta(G)>n/11-1=a-1$, then Lemma~\ref{lem:degree} applied with $t=10$ yields that $G$ is Hamiltonian. Assume henceforth that $\delta(G)\le a-1$.

The graph $G$ is not complete. Indeed, a complete graph would satisfy $\delta(G)=n-1\le a-1$, hence $n\le n/11$, which is false for $n\ge 3$. Lemma~\ref{lem:delta20} then yields $\delta(G)\ge 20$. Combining this lower bound with $\delta(G)\le a-1$ gives $a\ge 21$ and $n=11a\ge 231$.

Since $G$ is not complete and $n\ge 3$, one has $\alpha(G)\ge 2$. Let $I$ be a maximum independent set, so $|I|=\alpha(G)\ge 2$. Set $X:=V(G)\setminus I$. Then $c(G-X)=|I|\ge 2$, and $10$-toughness gives $|X|\ge 10|I|$, hence $|I|\le n/11=a$, and therefore $\alpha(G)\le a$.

The remainder of the argument splits on the threshold $4a=4n/11$.

\medskip
\noindent\textbf{Case 1.}
Some edge $uv$ of $G$ satisfies $|N_G(u)\cup N_G(v)|\le 4a$.

Apply Proposition~\ref{prop:small} to the edge $uv$. It supplies a cutset $S:=(N_G(u)\cup N_G(v))\setminus\{u,v\}$ such that $D_1:=G[\{u,v\}]$ is a component of $G-S$ and $c(G-S)=2$. Writing $D_2$ for the other component, one has $|V(D_2)|\ge 7a$ and $D_2$ is $(P_2 \cup P_1)$-free. Writing $S_1:=\{x\in S : |N_G(x)\cap V(D_2)|<2a\}$ and $G_1:=G[S_1\cup\{u,v\}]$, the graph $G_1$ is $(P_2 \cup P_1)$-free.

Write $S_2:=S\setminus S_1$ and $G_2:=G[S_2\cup V(D_2)]$. Then $V(G)=V(G_1)\cup V(G_2)$ is a partition, $|V(G_1)|\ge 2$, and $|V(G_2)|\ge 7a\ge 147\ge 2$. Proposition~\ref{prop:small} further supplies a $V(G_2)$-matched path-cover $Q$ of $V(G_1)$ with $|Q|=\max\{1,s(G_1)\}$.

We claim that $|Q|\le\alpha(G_1)\le a$. If $G_1$ is complete, then $|Q|=1=\alpha(G_1)$. If $G_1$ is not complete, then for every vertex set $X$ with $c(G_1-X)\ge 2$, choosing one vertex from each component of $G_1-X$ produces an independent set of $G_1$, hence $c(G_1-X)\le\alpha(G_1)$, and therefore $s(G_1)\le\alpha(G_1)$. Since $\alpha(G_1)\ge 1$, we obtain $|Q|=\max\{1,s(G_1)\}\le\alpha(G_1)\le\alpha(G)\le a$.

For each path $P\in Q$ write $z_P$ and $w_P$ for its two endpoints in $V(G_2)$. Let $L$ be the set of $|Q|$ edges $\{z_Pw_P : P\in Q\}$, retaining any such edge already present in $G_2$. Then $L$ is a matching of size $|L|=|Q|\le a$ with both ends in $V(G_2)$.

Let $G_2^*$ be obtained from $G_2$ by adding every edge of $L$. Proposition~\ref{prop:small} applied to this set $L$ yields $\kappa(G_2)\ge 2a$, $\kappa(G_2^*)\ge\kappa(G_2)\ge 2a$, and $\alpha(G_2^*)\le\alpha(G)\le a$, hence $\kappa(G_2^*)\ge |L|+\alpha(G_2^*)$. Moreover $|V(G_2^*)|\ge 7a\ge 147\ge 3$.

Lemma~\ref{lem:matching-cycle} applied to $H:=G_2^*$ and the matching $L$ therefore supplies a Hamilton cycle $C^*$ of $G_2^*$ containing every edge of $L$. Replace each edge $z_Pw_P\in L$ on $C^*$ by the corresponding path $P$ of $Q$. The paths of $Q$ are pairwise vertex-disjoint, cover $V(G_1)$, and meet $V(G_2)$ only at their endpoints, so the resulting closed walk is a Hamilton cycle of $G$. This completes Case~1.

\medskip
\noindent\textbf{Case 2.}
Every edge $uv$ of $G$ satisfies $|N_G(u)\cup N_G(v)|>4a$.

Define $S:=\{v\in V(G):\deg(v)<2a\}$, $S_1:=\{x\in S:\deg(x)<a\}$, and $\ell:=|S_1|$. Write $G_2:=G-S$.

The set $S$ is independent: an edge $xy$ inside $S$ would force $|N_G(x)\cup N_G(y)|>4a$ while $|N_G(x)\cup N_G(y)|\le\deg(x)+\deg(y)<4a$. Thus $|S|\le\alpha(G)\le a$.

The set $S_1$ is nonempty: $\delta(G)\le a-1$, so any vertex $v$ of degree $\delta(G)$ satisfies $v\in S_1$. In particular $\ell\ge 1$ and $\ell\le|S|\le a$.

Lemma~\ref{lem:star} applied to the noncomplete $10$-tough graph $G$ and the independent set $I:=S_1$ supplies a $K_{1,2}$-matching whose centres are exactly $S_1$. For each $x\in S_1$ write $z_x$ and $w_x$ for the two leaves. The $2\ell$ leaves are pairwise distinct and lie in $V(G_2)$.

Let $L$ be the set of $\ell$ edges $\{z_xw_x : x\in S_1\}$. Then $L$ is a matching of size $|L|=\ell$ with both ends in $V(G_2)$. Let $G_2^*$ be obtained from $G_2$ by adding every edge of $L$.

Proposition~\ref{prop:large} now applies with $S_1$ nonempty. It yields $|V(G_2)|\ge 10a$. We claim that $\kappa(G_2)\ge\ell+a$ in both subcases. The noncomplete subcase is Proposition~\ref{prop:large}(3). In the complete subcase, $\ell\le a$ and $a\ge 21$ give $\kappa(G_2)=|V(G_2)|-1\ge 10a-1\ge 2a\ge\ell+a$.

Adding edges cannot lower connectivity or raise the independence number, so $\kappa(G_2^*)\ge\kappa(G_2)\ge\ell+a$ and $\alpha(G_2^*)\le\alpha(G)\le a$. With $|L|=\ell$ one has $\kappa(G_2^*)\ge |L|+\alpha(G_2^*)$. Moreover $|V(G_2^*)|\ge 10a\ge 210\ge 3$.

Lemma~\ref{lem:matching-cycle} applied to $H:=G_2^*$ and the matching $L$ therefore supplies a Hamilton cycle $C^*$ of $G_2^*$ containing every edge of $L$. Replace each edge $z_xw_x\in L$ on $C^*$ by the path $z_x{-}x{-}w_x$. The stars are vertex-disjoint and their centres lie outside $V(G_2)$, so the result is a cycle $C_0$ of $G$ with $V(C_0)=V(G_2)\cup S_1$.

It remains to insert $R:=S\setminus S_1$. Start with the cycle $C:=C_0$. While $R$ is not contained in $V(C)$, choose $x\in R\setminus V(C)$ and put $H:=G[\{x\}]$. Then $C$ is non-Hamiltonian and $H$ is a connected subgraph of $G-V(C)$. Independence of $S$ and the membership $x\in S\setminus S_1$ give $x$ at least $a$ neighbours in $V(G_2)\subseteq V(C)$, hence strictly more than $a-1$ neighbours on $C$. Lemma~\ref{lem:cycle-ext} applied with $t=10$ therefore supplies a cycle $C'$ with $V(C)\subseteq V(C')$ and $x\in V(C')$. Replace $C$ by $C'$. After at most $|R|$ iterations one has $R\subseteq V(C)$. Since the cycle already contains $V(G_2)\cup S_1$, it is Hamiltonian. This completes Case~2.

Both cases produce a Hamilton cycle of $G$. This proves the theorem.
\end{proof}

\section{CONCLUDING REMARKS}

The proof of Theorem~\ref{thm:main} has two structural components. On the small-neighbourhood side, the inequalities
\[
|V(D_2)|\ge 7a
\qquad\text{and}\qquad
|N(\{x,y,z\})\cap V(D_2)|<6a
\]
leave more than $a$ vertices, while toughness gives $\alpha(G)\le a$; this forces the edge needed for the compression argument. On the large-neighbourhood side, Lemma~\ref{lem:asymm} analyses the two surviving components asymmetrically and yields
\[
\kappa(G_2)\ge \ell+a.
\]
This is the step that supplies the connectivity required by the prescribed-matching cycle theorem.

Lowering the toughness threshold below $10$ would require a stronger large-side connectivity mechanism and new estimates for the balanced path-cover branch. Developing such refinements remains an open direction.

\end{document}